\documentclass[wcp]{jmlr}

\usepackage{longtable}
\usepackage{booktabs}
\usepackage{enumitem}

\usepackage{booktabs}

\usepackage{lineno}

\makeatletter
\let \@jmlrpages \@empty
\makeatother

\jmlrvolume{353,}
\jmlryear{2026}
\jmlrworkshop{ACML 2026}

\newcommand{\bits}{\{0,1\}}
\newcommand{\EOS}{\mathtt{[EOS]}}
\newcommand{\blank}{\sqcup}
\newcommand{\leftend}{\vdash}
\newcommand{\ok}{\mathtt{ok}}
\newcommand{\readcmd}{\mathtt{read}}
\newcommand{\leftcmd}{\mathtt{left}}
\newcommand{\rightcmd}{\mathtt{right}}
\newcommand{\staycmd}{\mathtt{stay}}
\newcommand{\writecmd}[1]{\mathtt{write}_{#1}}

\makeatletter
\@ifundefined{theorem}{\newtheorem{theorem}{Theorem}}{}
\@ifundefined{lemma}{\newtheorem{lemma}{Lemma}}{}
\@ifundefined{proposition}{\newtheorem{proposition}{Proposition}}{}
\@ifundefined{corollary}{\newtheorem{corollary}{Corollary}}{}
\@ifundefined{definition}{\newtheorem{definition}{Definition}}{}
\@ifundefined{remark}{\newtheorem{remark}{Remark}}{}
\@ifundefined{conjecture}{}{}
\makeatother

\title[When Does Tool Use Increase the Expressive Power of Finite-Precision Recurrent Models?]{When Does Tool Use Increase the Expressive Power of Finite-Precision Recurrent Models?}

 \author{\Name{Nikola Zubić} \Email{zubic@ifi.uzh.ch}\\
 \addr Robotics and Perception Group, University of Zurich
 \AND
 \Name{Qian Li} \Email{}\\
 \addr Shenzhen International Center for Industrial and Applied Mathematics, \\Shenzhen Research Institute of Big Data
 \AND
 \Name{Yuyi Wang} \Email{}\\
 \addr Tengen Intelligence Institute, CRRC Zhuzhou Institute
 \AND
 \Name{Davide Scaramuzza} \Email{}\\
 \addr Robotics and Perception Group, University of Zurich}

\editors{Andy Song, Bo Han and Sarah Erfani}

\begin{document}

\makeatletter
\let \@jmlrpages \@empty
\makeatother

\maketitle

\begin{abstract}
Modern sequence models are increasingly deployed as agents that
interleave token generation with calls to external tools. We study
when such tool access increases computational expressivity. We model
any fixed finite-precision recurrent sequence model, including
finite-precision state-space models (SSMs) with $B$ bits of internal
state, as a deterministic finite-state controller interacting with an
oracle through a finite command/observation interface, and establish a
separation between the two regimes analyzed. First, every finite-state
bounded-interface tool with memory set $M$ can be internalized at a
cost of at most $\lceil\log_2|M|\rceil+O(1)$ recurrent bits, so the
combined system remains finite-state. Second, a single unbounded tape
with only local $\mathtt{read}$/$\mathtt{write}$/$\mathtt{move}$
commands suffices for Turing completeness: every single-tape Turing
machine with state set $Q$ and tape alphabet $\Gamma$ is simulated by a
controller with $O(\log|Q|+\log|\Gamma|)$ recurrent bits, and
$\mathrm{EQ}_n$ requires $2^n$ states without tools but one
constant-size controller with the tape. Third, this construction is
realized exactly by a one-layer finite-precision \emph{selective}
affine SSM with binary one-hot states, $\{0,1\}$ transition matrices,
and zero biases. Under explicitly
stated resource conventions, $O(\log B)$ recurrent bits suffice for
tape-assisted simulation of any $B$-state Turing machine, whereas
directly realizing an arbitrary $B$-state transition map in one exact
affine update without external memory requires total recurrent
dimension exactly $B-1$. Fixed description size, occupied tape cells,
and interaction cost are reported separately from recurrent memory. Our
results concern exact expressivity under deterministic semantics, not
learnability or empirical performance.
\end{abstract}

\begin{keywords}
expressivity; state-space models; tool use; oracle machines; Turing
completeness; finite precision; recurrent models; computational
complexity
\end{keywords}

\section{Introduction}\label{sec:intro}

Sequence models are no longer used only as sequence-to-sequence maps:
they are deployed as \emph{agents} that interleave output generation
with calls to external tools such as calculators, code interpreters,
retrieval systems, and scratchpads
\citep{schick2023toolformer,yao2023react,nye2021scratchpad}. This
raises a basic theoretical question:
\begin{center}
\emph{Which tools increase the computational expressivity of a fixed
finite-precision sequence model, by how much, and at what internal
memory cost?}
\end{center}

The starting point is that any fixed finite-precision recurrent model
is, from the standpoint of expressivity, a finite-state machine: a
model with $B$ bits of internal state has at most $2^B$
configurations, and hence, without external interaction, recognizes
only regular languages
\citep{weiss2018practical,merrill2019automata,zubic2024limits}. This
covers, in particular, finite-precision structured state-space models
(SSMs) \citep{gu2022s4,gu2023mamba}, whose formal-language limitations
have been analyzed in detail
\citep{merrill2024illusion,sarrof2024expressive,zubic2024limits,zubic2026ssm}.
Chain-of-thought decoding changes this picture only partially: it can
add power, but in a way that is tightly constrained by the number of
generated steps
\citep{merrill2024cot,zubic2026ssm}. Tool use is a different
augmentation mechanism. In the two regimes we analyze, we show that
whether it increases expressivity is determined by the \emph{state
space and interface of the tool}, not by the controller.

\paragraph{Our model.}
We abstract the sequence model as a deterministic finite-state
controller that first reads its input and then, in a generation phase,
repeatedly emits output tokens, issues tool commands from a finite
command alphabet, receives observations from a finite observation
alphabet, or halts (Section~\ref{sec:finite-state-controller}). Any
fixed finite-precision recurrent model with $B$ bits of state induces
such a controller with at most $2^B$ states, so all our lower bounds
apply to such models, and all our upper bounds are constructive.

\paragraph{Contributions.}
The basic ingredients are classical: product-state composition of
finite automata and finite-control simulation of a Turing machine
through a tape are textbook constructions, and one-hot realizations of
finite automata are closely related to prior emulation results for
recurrent networks and SSMs
\citep{siegelmann1995computational,merrill2019automata,merrill2024illusion,sarrof2024expressive}.
Our contribution is their integration into a single formal semantics
for tool-augmented finite-precision recurrent models, together with
the following interface-sensitive and quantitative results
(Table~\ref{tab:summary}). Logically, we prove a \emph{universal
internalization result} for finite-state bounded-interface tools and a
\emph{sufficiency result} for one specific unbounded tape. Our contributions are as follows:
\begin{enumerate}[leftmargin=2em]
\item \textbf{Finite-state tools add (almost) nothing.}
Proposition~\ref{prop:product} treats \emph{every} finite-state
bounded-interface tool (Definition~\ref{def:fs-oracle}) and gives an
explicit internalization cost of at most
$\lceil\log_2|M|\rceil+O(1)$ recurrent bits, where $M$ is the oracle
memory. Definition~\ref{def:fs-oracle} requires not only finite tool
memory but also finite-state input initialization and finite
command/observation interfaces. Remark~\ref{rem:init} shows why: an
arbitrary initializer or whole-command transition could hide
nonregular or even undecidable computation in a nominally
finite-memory tool.
\item \textbf{One minimal unbounded tool suffices for Turing
completeness.} For every deterministic single-tape Turing machine $T$
in output-normal form, Theorem~\ref{thm:tape-turing} constructs a
$T$-specific controller with $O(\log|Q|+\log|\Gamma|)$ recurrent bits
that simulates $T$ through a local read/write/move tape
(Definition~\ref{def:tape}). Separately, Theorem~\ref{thm:eq} gives
one fixed tape-assisted controller solving every $\mathrm{EQ}_n$,
while any no-tool controller solving $\mathrm{EQ}_n$ for a fixed $n$
needs at least $2^n$ states.
\item \textbf{Exact realization by a one-layer selective SSM.}
Lemma~\ref{lem:ssm-onehot} and Theorem~\ref{thm:ssm-turing-complete}
lift the finite-control construction to an exact one-layer
\emph{selective affine} SSM: binary one-hot hidden states, $\{0,1\}$
transition matrices, zero biases, the fixed finite action readout of
Definition~\ref{def:sel-ssm}, and hidden dimension $O(|Q|\,|\Gamma|)$.
This is specifically a result for observation-selective transitions.
\item \textbf{Resource accounting (Section~\ref{sec:resources},
proofs in the supplement).} For a $B$-state Turing machine over a
fixed tape alphabet, tape-assisted simulation uses $O(\log B)$
recurrent bits when the machine-specific fixed parameters/readout and
the occupied tape cells are outside that budget
(Theorem~\ref{thm:log-main}). Without external memory, the worst-case
minimum total recurrent dimension for directly realizing a map
$f:[B]\times\{0,1\}\to[B]$ in one exact affine update is exactly
$B-1$ under a triangular multilayer selective-affine convention
(Theorem~\ref{thm:direct-main}). We also make the three-symbol
one-hot construction fully explicit with constant $C_0=9$
(Appendix~\ref{app:three-symbol}).
\end{enumerate}

\paragraph{Roadmap.}
Section~\ref{sec:finite-state-controller} defines the controller and
the resource conventions. The argument then proceeds in four steps:
finite-tool internalization (Proposition~\ref{prop:product},
Section~\ref{sec:finite-state-tools}), tape-assisted universality
(Theorem~\ref{thm:tape-turing}, Sections~\ref{sec:tape-oracle}
and~\ref{sec:interactive-tape}), the finite exponential separation
(Theorem~\ref{thm:eq}, Section~\ref{sec:eq-lower-bound}), and exact
selective-SSM realization (Lemma~\ref{lem:ssm-onehot} and
Theorem~\ref{thm:ssm-turing-complete},
Section~\ref{sec:ssm-controller}). Section~\ref{sec:resources} states
the logarithmic and direct-realization results, whose proofs are in
the appendix.

\begin{table}[htbp]
\centering
\caption{Summary of results. ``Recurrent memory'' is the dynamic
internal state of the controller, $\lceil\log_2|S|\rceil$ bits for a
state set $S$, or hidden dimension $d$ at fixed per-coordinate
precision for SSMs. Fixed description (parameters/readout), occupied
tape cells, and interaction cost are outside this metric and are
reported separately (Section~\ref{sec:finite-state-controller}). No
time-bound is claimed. $Q,\Gamma$: simulated Turing machine states
and tape alphabet, $M$: oracle memory, $B$: number of simulated
states.}\label{tab:summary}
\resizebox{\textwidth}{!}{%
\begin{tabular}{@{}lllll@{}}
\toprule
Tool & Recurrent memory & Outside the metric & Power & Where \\
\midrule
none & $b$ bits & --- & regular languages & Sec.~\ref{sec:interactive-tape} \\
finite-state oracle & $+\lceil\log_2|M|\rceil+O(1)$ bits & --- & still regular & Sec.~\ref{sec:finite-state-tools} \\
tape oracle & $O(\log|Q|+\log|\Gamma|)$ bits & tape cells & Turing complete & Sec.~\ref{sec:interactive-tape} \\
tape oracle (1-layer sel.\ SSM) & $d=O(|Q|\,|\Gamma|)$, binary & tape cells, readout & Turing complete & Sec.~\ref{sec:ssm-controller} \\
none, $\mathrm{EQ}_n$ & $\geq 2^n$ states & --- & --- & Sec.~\ref{sec:eq-lower-bound} \\
tape oracle, universal & $O(\log B)$ bits & tape cells, readout of size $O(B\log B)$ & simulates $\mathcal T_B$ & Sec.~\ref{sec:resources}, App.~\ref{app:log-simulation} \\
no external memory, direct & exactly $B-1$ dims & --- & one-step $f$ & Sec.~\ref{sec:resources}, App.~\ref{app:log-simulation} \\
\bottomrule
\end{tabular}}
\end{table}

\section{Related Work}\label{sec:related}

\paragraph{Expressivity of finite-precision recurrent models and SSMs.}
With unbounded precision and time, recurrent networks are Turing
complete \citep{siegelmann1995computational}, but under realistic
finite precision they collapse to finite-state machines
\citep{weiss2018practical,merrill2019automata}. For SSMs
\citep{gu2022s4,gu2023mamba}, circuit-complexity upper bounds place
fixed-depth models in $\mathsf{TC}^0$ and expose state-tracking
limitations \citep{merrill2024illusion}, with refined formal-language
characterizations by \citet{sarrof2024expressive} and expressivity
gains from richer recurrences \citep{grazzi2025unlocking}. Several of
these works emulate finite automata with recurrent or SSM layers: our
one-hot realization (Lemma~\ref{lem:ssm-onehot}) is of the same
classical kind, and we use it as a building block rather than claim it
as new. \citet{zubic2024limits} prove that finite-precision SSMs
recognize only regular languages and establish function-composition
lower bounds for one-layer SSMs, supported empirically on
compositional benchmarks. Related compositional limitations of
Transformers appear in \citet{dziri2023faith} and
\citet{peng2024limitations}. \citet{zubic2026ssm} analyze multi-layer
SSMs, showing an inherent gap to streaming models, that offline
chain-of-thought does not fundamentally increase expressiveness while
online chain-of-thought does, and that width and precision are
non-interchangeable in the base model. Our work is complementary: we
fix the controller to be finite-state (subsuming any fixed
finite-precision model, including those studied above) and vary the
\emph{tool interface}, obtaining a separation between finite-state
bounded-interface tools and a tape-like tool.

\paragraph{Augmented decoding and Turing completeness.}
Turing-completeness results for Transformers typically require
unbounded precision or growing intermediate computation
\citep{perez2021attention}, and chain-of-thought provably extends
expressive power as a function of the number of generated steps
\citep{merrill2024cot,wei2022cot}. The neural Chomsky-hierarchy study
of \citet{deletang2023chomsky} demonstrates empirically that external
memory structure (stack, tape) determines which language classes are
learnable. Simulating a Turing machine by a finite control with a
tape is, of course, classical. Our tape-oracle theorem restates it at
the level of tool-augmented finite-state controllers with explicit
recurrent-memory bounds, and our appendix results quantify how much
recurrent memory is relocated to the tape and the fixed readout, under
stated conventions (from $\Theta(B)$ to $O(\log B)$, using small
universal machines \citep{NearyWoods2009}).

\paragraph{Tool-augmented language models.}
Practical systems interleave generation with tool calls
\citep{schick2023toolformer,yao2023react} or scratchpad writes
\citep{nye2021scratchpad}. These works are empirical. We provide a
minimal formal semantics for the command/observation loop (tags such
as $\mathtt{[TOOL]}$/$\mathtt{[OBS]}$ being a finite-state encoding
convention) and determine, for the two classes of tools analyzed,
whether the loop increases expressivity.

\section{A finite-state controller model}\label{sec:finite-state-controller}

We work with finite alphabets throughout. Multi-token tags such as
$
\mathtt{[TOOL]}, \mathtt{[\backslash TOOL]},
\mathtt{[OBS]}, \mathtt{[\backslash OBS]}
$
are only a finite-state encoding convention: they can be replaced by a
finite command alphabet and a finite observation alphabet without
changing the asymptotic conclusions, and the literal tagged transcript
can be recovered by expanding the finite control by a constant factor.

Let \(A\) be an input alphabet and let \(\dashv\notin A\) be a right endmarker.

\begin{definition}[Finite-state controller]\label{def:controller}
A deterministic finite-state controller is a finite-state transducer
$
C=(S,s_0,\delta_{in},\alpha,\delta_{out},\delta_{obs}),
$
where:
\begin{itemize}[leftmargin=2em]
\item \(S\) is a finite state set;
\item \(s_0\in S\) is the initial state;
\item \(\delta_{in}:S\times (A\cup\{\dashv\})\to S\) updates the controller while reading the input;
\item after the input phase, \(\alpha\) determines the next action of the controller. The possible actions are:
\[
\mathrm{halt},\qquad \mathrm{out}(a)\ \text{for an output token }a,\qquad
\mathrm{cmd}(d)\ \text{for a tool command }d;
\]
\item \(\delta_{out}\) updates the controller after an output token is emitted;
\item \(\delta_{obs}\) updates the controller after a tool observation is received.
\end{itemize}
The controller first reads \(x\dashv\) from left to right. It then enters the generation phase. In the generation phase it repeatedly outputs tokens, issues tool commands, or halts.
\end{definition}

\paragraph{Example (a tool call and its tagged transcript).}
Let \(D=\{\mathtt{parity}\}\) be a one-command tool interface with
observation alphabet \(R=\{\mathtt{even},\mathtt{odd}\}\). A
controller may have a post-input state \(s_{ask}\) with
$\alpha(s_{ask})=\mathrm{cmd}(\mathtt{parity})$. As a literal
transcript, this abstract command can be displayed as the fixed string
$\mathtt{[TOOL]}\;\mathtt{parity}\;\mathtt{[\backslash TOOL]}$. The
environment returns an observation \(r\in R\), e.g.\
$\mathtt{[OBS]}\;\mathtt{odd}\;\mathtt{[\backslash OBS]}$, and the
controller updates by $s'=\delta_{obs}(s_{ask},\mathtt{odd})$. For
instance, \(\alpha(s')=\mathrm{out}(\mathtt{F})\), after which
\(\delta_{out}\) moves to a halting state. The tags are not additional
sources of computational power but only a fixed, finite encoding
of the abstract command and observation alphabets.

\begin{remark}
This is an explicit finite-state version of the bounded-memory sequence controller. Any fixed finite-precision recurrent model, including a finite-precision state-space model with \(B\) bits of internal state, induces such a controller with at most \(2^B\) possible internal states.
\end{remark}

\paragraph{Resource accounting.}
For a controller with state set $S$ we report \emph{recurrent
(dynamic internal) memory} as $\lceil\log_2|S|\rceil$ bits. For SSM
realizations we separately report the hidden dimension $d$ at fixed
per-coordinate precision. Three further resources lie \emph{outside
the recurrent-state metric and are reported separately}:
(i)~\emph{fixed description}: the transition/readout description,
which may scale with the simulated machine (in the $O(\log B)$
construction of Section~\ref{sec:resources} it encodes the simulated
machine, and only the recurrent-state cost, not the total description
size, is logarithmic);
(ii)~\emph{external work memory}: occupied tape cells, which may grow
with input length or running time and are included in any total-space
accounting (our tape results \emph{relocate} persistent working memory
from recurrent state to external storage, they do not eliminate it);
(iii)~\emph{interaction cost}: no time or communication bound is
claimed. Recurrent steps and tool calls scale with the simulated
computation and are unbounded over inputs.

\section{Finite-state tools do not increase expressivity}\label{sec:finite-state-tools}

Finite-state tools act as additional bounded-memory devices, but only
if the tool \emph{interface} itself is finite-state. Merely requiring
the oracle memory set \(M\) to be finite is not enough: an arbitrary
initialization map or an arbitrary command-to-observation map could
encode nonregular or even undecidable behavior. We therefore use the
following definition.

\begin{definition}[Finite-state bounded-interface oracle]\label{def:fs-oracle}
A finite-state bounded-interface oracle is a tuple
$
\mathcal O=(M,m_0,\eta,D,R,\mathrm{Step}),
$
where:
\begin{itemize}[leftmargin=2em]
\item \(M\) is a finite oracle-memory set;
\item \(m_0\in M\) is the initial oracle state before reading the input;
\item \(\eta:M\times A\to M\) is a finite-state initialization transition;
\item \(D\) is a finite command set;
\item \(R\) is a finite observation set;
\item
$
\mathrm{Step}:M\times D\to M\times R
$
is a deterministic transition-and-observation map.

Operationally, when the current oracle memory is \(m\) and the
controller issues a command \(d\in D\), the value
$
    \mathrm{Step}(m,d)=(m',r)
$
means that the oracle memory becomes \(m'\) and the observation returned
to the controller is \(r\in R\).
\end{itemize}
Given input \(x=x_1\cdots x_n\in A^\ast\), the initial oracle memory is
$
\mathrm{Init}(x)=\eta^\ast(m_0,x),
$
where \(\eta^\ast\) is the usual extension of \(\eta\) to words.
\end{definition}

\paragraph{Example (a one-bit oracle).}
Take $A=\{0,1\}$ and let the oracle memory be $M=\{0,1\}$, with
$m_0=0$ and initialization transition $\eta(m,a)=m\oplus a$.  After
reading the input, the oracle state records the parity of the input.
Let the command set be $D=\{\mathtt{query},\mathtt{flip}\}$ and the
observation set be $R=\{0,1,\mathtt{ok}\}$.  Define
\[
    \mathrm{Step}(m,\mathtt{query})=(m,m),
    \qquad
    \mathrm{Step}(m,\mathtt{flip})=(1-m,\mathtt{ok}).
\]
This is a finite-state bounded-interface oracle: both its memory and
its command/observation interface are finite.  The product-state
simulation below simply stores the pair consisting of the controller
state and this one-bit oracle state.

\begin{proposition}[Product-state simulation]\label{prop:product}
Let \(C\) be a deterministic finite-state controller interacting with a finite-state bounded-interface oracle \(\mathcal O\). Then there is a deterministic finite-state controller \(\widetilde C\), without tool access, that produces exactly the same scored output stream on every input.
Moreover, \(\widetilde C\) may be chosen with state space contained in
$
S\times M
$
up to a constant-factor enlargement for fixed command and observation encodings. Thus finite-state bounded-interface tools increase recurrent memory by at most
$(
\lceil\log_2 |M|\rceil+O(1)
)$
bits.
\end{proposition}

\begin{proof}
Let
$
C=(S,s_0,\delta_{in},\alpha,\delta_{out},\delta_{obs})
$
and let
$
\mathcal O=(M,m_0,\eta,D,R,\mathrm{Step}).
$
Construct a no-tool controller \(\widetilde C\) whose state records a pair
$
(s,m)\in S\times M.
$

During the input phase, when the next input symbol is \(a\in A\), define
$
(s,m)\mapsto \bigl(\delta_{in}(s,a),\\ \eta(m,a)\bigr).
$
On the endmarker \(\dashv\), update the controller component by \(\delta_{in}\) and leave the oracle component fixed.

During the generation phase, suppose \(\widetilde C\) is in state \((s,m)\).
If \(\alpha(s)=\mathrm{halt}\), then \(\widetilde C\) halts.
If \(\alpha(s)=\mathrm{out}(a)\), then \(\widetilde C\) emits the same output token \(a\) and moves to
$
(\delta_{out}(s,a),m).
$
If \(\alpha(s)=\mathrm{cmd}(d)\), compute
$
(m',r)=\mathrm{Step}(m,d).
$
The original interaction would update the oracle memory from \(m\) to \(m'\), return the observation \(r\), and update the controller state to \(\delta_{obs}(s,r)\). Therefore \(\widetilde C\) moves to
$
(\delta_{obs}(s,r),m').
$
No tool call is made externally, its effect has been simulated internally.

By induction on the number of generation steps, the simulated pair \((s,m)\) is exactly the controller state and oracle state that would occur in the true interaction. Hence every output token emitted by \(\widetilde C\) is exactly the corresponding output token emitted by \(C\), and conversely. Thus the scored output streams agree on every input.
If one insists on reproducing the literal tagged transcript rather than only the scored output stream, one adds a finite emission buffer for the fixed encodings of \(d\in D\) and \(r\in R\). Since \(D\) and \(R\) are finite, this changes the state space only by a constant factor.
\end{proof}

\begin{remark}\label{rem:init}
The finite-state initialization hypothesis is essential. If \(\mathrm{Init}:A^\ast\to M\) is allowed to be arbitrary, then even a one-bit oracle state can encode membership in an arbitrary language. Similarly, if \(\mathrm{Step}:M\times \Sigma^\ast\to M\times \Sigma^\ast\) is arbitrary as a function of the whole command string, then the oracle can perform unbounded computation despite having finite memory. The product-state simulation theorem is therefore valid only for genuinely finite-state tool interfaces.
\end{remark}

In resource terms, a tool satisfying Definition~\ref{def:fs-oracle}
with fixed finite capacity therefore costs a multiplicative $|M|$
state blow-up (up to the fixed interface factor) when internalized. It may change efficiency, but never the class of computable
behaviors.

\section{An unbounded tape-memory tool}\label{sec:tape-oracle}

We now define a simple infinite-state tool. It performs only local tape operations, but its memory is unbounded.
Let \(\Gamma\) be a finite tape alphabet containing a blank symbol \(\blank\), a left-end marker \(\leftend\), and an output endmarker \(\EOS\). Let \(A\subseteq \Gamma\setminus\{\blank,\leftend,\EOS\}\) be the input alphabet.

A tape configuration is a pair \((T,h)\), where
$
T:\mathbb N\to \Gamma
$
has finite support away from blanks, satisfies \(T(0)=\leftend\), and \(h\in\mathbb N\) is the head position.

\begin{definition}[Tape oracle]\label{def:tape}
The tape oracle \(\mathcal O_{tape}\) has memory space consisting of all tape configurations \((T,h)\). On input \(x=x_1\cdots x_n\in A^\ast\), it initializes
\[
T(0)=\leftend,\qquad T(i)=x_i\ \text{for }1\le i\le n,\qquad
T(i)=\blank\ \text{for }i>n,
\]
and sets \(h=1\).

It supports the finite command set
\[
D=\{\mathtt{read},\mathtt{move\_left},\mathtt{move\_right}\}
   \cup \{\mathtt{write}(\gamma):\gamma\in\Gamma\}.
\]
The observations are symbols of \(\Gamma\), returned by \(\mathtt{read}\), together with a dummy acknowledgement symbol \(\mathtt{ok}\) for write and move commands.

The transition rules are:
\begin{itemize}[leftmargin=2em]
\item \(\mathtt{read}\) returns \(T(h)\) and does not change the tape;
\item \(\mathtt{write}(\gamma)\) writes \(\gamma\) at the current head position, except that the left-end marker at cell \(0\) is preserved;
\item \(\mathtt{move\_right}\) replaces \(h\) by \(h+1\);
\item \(\mathtt{move\_left}\) replaces \(h\) by \(\max\{h-1,0\}\).
\end{itemize}
\end{definition}

\paragraph{Example (local tape operations).}
On input $x=01$, the initial tape begins
$
    \leftend\;0\;1\;\blank\;\blank\;\cdots,
$
and the head is on the first input symbol.  The command
$\mathtt{read}$ returns $0$.  After $\mathtt{move\_right}$, the head is
on the symbol $1$. Issuing $\mathtt{write}(0)$ changes the tape prefix
to $\leftend\,0\,0$.  Moving right again and writing $1$ uses the next
previously blank cell.  Although every command is local and the
command alphabet is finite, an arbitrarily long computation can visit
and modify arbitrarily many tape cells, which is why the oracle has an
infinite state space. The usable capacity of the tape, and hence the
occupied external memory, may grow with the computation.

\section{Interactive tape access yields Turing completeness}\label{sec:interactive-tape}

We use the following standard output convention. A deterministic single-tape Turing machine is in output-normal form if, whenever it halts, its tape contains
"$
\leftend\, y\, \EOS
$"
starting at cell \(0\), the head is positioned at the first symbol of \(y\EOS\), and \(y\in A^\ast\) is the output. Every partial computable string function has a deterministic single-tape machine in this normal form: after computing its intended output, the machine can rewrite the tape into the required form and move its head to the first output cell.

\begin{theorem}[Interactive tape tool-use is Turing complete]\label{thm:tape-turing}
Let \(T\) be a deterministic single-tape Turing machine in output-normal form, with finite state set \(Q\) and tape alphabet \(\Gamma\). There exists a deterministic finite-state controller \(C_T\) such that, with interactive access to \(\mathcal O_{tape}\), the controller computes the same partial function as \(T\).
More precisely, for every input \(x\in A^\ast\), the controller halts and outputs \(T(x)\EOS\) if and only if \(T\) halts on \(x\). The number of controller states is \(O(|Q||\Gamma|)\), equivalently
$
O(\log |Q|+\log |\Gamma|)
$
bits of internal memory. If \(\Gamma\) is fixed, this is \(O(\log |Q|)\) bits and is constant with respect to the input length.
\end{theorem}

\begin{proof}
Let the transition function of \(T\) be
$
\delta_T:Q_{nh}\times \Gamma
\to Q\times \Gamma\times \{L,R,S\},
$
where \(Q_{nh}\subseteq Q\) is the set of nonhalting states.
The controller \(C_T\) stores:
\begin{itemize}[leftmargin=2em]
\item the current simulated machine state \(q\in Q\);
\item a mode flag indicating whether it is about to read, write, move, or print;
\item when needed, one tape symbol \(\gamma\in\Gamma\) obtained from the most recent read.
\end{itemize}
This is a finite amount of information, with \(O(|Q||\Gamma|)\) possible states.
The simulation loop is as follows. Suppose the controller stores the simulated state \(q\in Q_{nh}\).
First, it issues the command \(\mathtt{read}\). If the observation is \(\gamma\in\Gamma\), it computes in finite control
$
\delta_T(q,\gamma)=(q',\gamma',d),
$
where \(d\in\{L,R,S\}\). It then issues the command \(\mathtt{write}(\gamma')\). If \(d=L\), it issues \(\mathtt{move\_left}\), if \(d=R\), it issues \(\mathtt{move\_right}\), if \(d=S\), it issues no move command. It then updates its stored simulated state to \(q'\).

The invariant is that after each completed simulated transition, the oracle tape configuration is exactly the tape configuration of \(T\), and the controller stores exactly the current finite control state of \(T\). The invariant holds at initialization by the definition of \(\mathcal O_{tape}\). The preceding paragraph preserves it because the controller performs exactly the read, write, and head-move operations specified by \(\delta_T\).

If \(q'\) is nonhalting, the controller repeats the loop. If \(q'\) is halting, the output-normal-form convention guarantees that the oracle head is positioned at the first output symbol and that the tape contains a contiguous word \(y\EOS\). The controller enters print mode. In print mode it repeatedly issues \(\mathtt{read}\). If the observed symbol is \(\EOS\), it emits \(\EOS\) as an output token and halts. Otherwise it emits the observed symbol as an output token and then issues \(\mathtt{move\_right}\).

Thus the emitted output stream is exactly \(y\EOS\), where \(y=T(x)\). If \(T\) never halts, the simulation loop never reaches print mode, so the controller also never halts. Therefore the controller computes the same partial function as \(T\).
\end{proof}

The bound in Theorem~\ref{thm:tape-turing} counts recurrent memory
only. The occupied tape cells and the number of tool calls scale with
the simulated computation and are outside this metric
(Section~\ref{sec:finite-state-controller}).

\begin{corollary}[Strict expressivity increase]
There are decision and search problems solvable by a finite-state controller with interactive access to \(\mathcal O_{tape}\) that are not solvable by any finite-state controller without tool access.
\end{corollary}

\begin{proof}
By the theorem, interactive access to \(\mathcal O_{tape}\) allows a finite-state controller to simulate any Turing machine, and hence to compute any partial computable string function.

Without tools, a controller with a fixed finite state set recognizes only regular languages under the usual decision convention: after reading \(x\dashv\), its final behavior is determined by one of finitely many states. Therefore the accepted language is a regular language.

The language
$
L_{eq}=\{u\#u:u\in \bits^\ast\}
$
is not regular. Hence, it is not decidable by any no-tool finite-state controller, but it is decidable by a finite-state controller with interactive tape access, since it is Turing-decidable.
\end{proof}

\section{A concrete exponential lower bound for equality}\label{sec:eq-lower-bound}

We now give a finite-length version of the same separation.

\begin{theorem}[Equality requires exponentially many internal states]\label{thm:eq}
For each \(n\in\mathbb N\), define the promise problem \(\mathrm{EQ}_n\) on inputs
$u\#v,$ where $u,v\in \bits^n,$ by requiring output \(\mathtt T\) if \(u=v\) and output \(\mathtt F\) otherwise. Then:

\begin{enumerate}[label=(\roman*), leftmargin=2em]
\item Any deterministic no-tool finite-state controller that solves \(\mathrm{EQ}_n\) on all promised inputs has at least \(2^n\) internal states.
\item There is a single deterministic finite-state controller which, with interactive access to \(\mathcal O_{tape}\), solves \(\mathrm{EQ}_n\) for every \(n\).
\end{enumerate}
\end{theorem}

\begin{proof}
\emph{Proof of (i).}
Let \(C\) be a deterministic no-tool finite-state controller solving \(\mathrm{EQ}_n\), and let \(S\) be its state set. After reading the first block \(u\in\bits^n\), the controller is in some state
$
\phi(u)\in S.
$
Thus \(\phi:\bits^n\to S\).
If \(|S|<2^n\), then there exist distinct \(u,u'\in\bits^n\) such that
$
\phi(u)=\phi(u').
$
Now feed both computations the same continuation \(\#u\dashv\). Since the controller is deterministic and starts this continuation in the same state in both cases, its subsequent behavior is identical on
$
u\#u
\text{ and }
u'\#u.
$
Therefore it must output the same answer on both inputs. But \(u\#u\) is a positive instance, while \(u'\#u\) is a negative instance because \(u'\ne u\). This contradicts correctness. Hence
$
|S|\ge 2^n.
$

\emph{Proof of (ii).}
The language
$
L_{eq}=\{u\#u:u\in\bits^\ast\}
$
is decidable by a deterministic Turing machine. For example, a single-tape machine can repeatedly mark the leftmost unmarked bit before the delimiter, remember whether it was \(0\) or \(1\), scan to the corresponding leftmost unmarked bit after the delimiter, compare and mark it, and then return to the beginning. When no unmarked bit remains before the delimiter, it checks that no unmarked bit remains after the delimiter. It accepts exactly when all comparisons succeed.

This algorithm uses only finitely many tape symbols and finitely many control states. By the Turing-completeness theorem above, there is a finite-state controller with interactive access to \(\mathcal O_{tape}\) that simulates this machine. The same controller works for all \(n\), because it decides the full language \(L_{eq}\), and therefore solves each promised problem \(\mathrm{EQ}_n\).
\end{proof}

\section{Turing Completeness with a Generalized SSM-Type Controller}\label{sec:ssm-controller}

We now show that the preceding Turing-completeness statement can be realized by a natural finite-precision SSM-type controller. The precise statement is architecture-dependent. We prove the positive result for \emph{selective affine} state-space controllers, i.e. controllers whose transition matrix may depend on the current observation token. This selectivity is essential to the proof: a fixed linear time-invariant recurrence with one common transition matrix is not claimed to realize arbitrary finite control.
The proof has two steps. First, we show that every deterministic finite-state controller is exactly realizable by a one-hot selective affine SSM. Second, we instantiate this construction with the finite controller that simulates a Turing machine using the read/write/move tape oracle.

\subsection{A generalized finite-precision SSM model}

Let \(\Omega\) be a finite observation alphabet. Let \(\mathbb F_p\subset \mathbb R\) denote a fixed finite set of real numbers representable with \(p\) bits, and assume \(0,1\in \mathbb F_p\).
A generalized \(L\)-layer finite-precision SSM controller has hidden states:
\[
    h_{\ell,t}\in \mathbb F_p^d,
    \qquad
    1\leq \ell\leq L,\ t\geq 0,
\]
and layer outputs
$
    y_{\ell,t}\in \mathbb F_p^m.
$
Given an observation token \(\omega_t\in\Omega\), set
$
    y_{0,t}=\operatorname{emb}(\omega_t,t),
$
where \(\operatorname{emb}\) is any fixed finite-precision embedding. The layers evolve according to
\[
    h_{\ell,t}
    =
    A_{\ell,t} h_{\ell,t-1}
    +
    B_{\ell,t} y_{\ell-1,t},
    \qquad
    y_{\ell,t}
    =
    \operatorname{out}_{\ell,t}
    \bigl(h_{\ell,t},y_{\ell-1,t}\bigr),
    \tag{SSM}
\]
where \(A_{\ell,t}\) and \(B_{\ell,t}\) have finite-precision entries. As in selective SSMs, the matrices and readout maps may depend on the layer \(\ell\), the time \(t\), and the current lower-layer input \(y_{\ell-1,t}\), equivalently on the current observation token \(\omega_t\) after embedding.
The construction below uses only the one-layer special case \(L=1\), with \(B_{1,t}=0\), binary hidden states, and transition matrices selected by the current observation.

\subsection{Interaction and finite-control conventions}

We reuse the tape oracle of Section~\ref{sec:tape-oracle}.  For this
section only, we write its commands as
\[
    \mathcal C
    =
    \{\mathtt{read},\mathtt{left},\mathtt{right},\mathtt{stay}\}
    \cup
    \{\mathtt{write}_{\tau}:\tau\in\Gamma\},
\]
where $\mathtt{stay}$ leaves the head in place.  We write
$\Omega=\Gamma\cup\{\bot\}$ for the observation alphabet: a read
returns the scanned tape symbol, while every write or move command
returns the dummy observation $\bot$.  If
$\Delta\subseteq\Gamma\setminus\{\blank,\leftend,\EOS\}$ is the
visible output alphabet, then
\[
    \mathsf{Act}
    =
    \mathcal C\cup\Delta\cup\{\EOS,\mathtt{halt}\}.
\]
The actions in $\mathcal C$ are tool commands.  The actions in
$\Delta\cup\{\EOS\}$ are visible output tokens, not tool commands.  When
such a visible token is emitted, no tape operation is performed. The
environment appends the token to the visible output stream and returns
the dummy observation $\bot$ only to advance the recurrent update.

The action alphabet $\mathsf{Act}$ does not include chain-of-thought
tokens, private scratchpad tokens, or other hidden reasoning tokens.
Any private computation of the controller is represented by its finite
control or recurrent state, not by emitted actions.  The visible output
stream retains only the symbols in $\Delta\cup\{\EOS\}$.

After the input has initialized the tape, the finite control used in
the simulation is the following Moore-machine specialization of the
controller from Section~\ref{sec:finite-state-controller}.

\begin{definition}[Observation-driven finite-state controller]
An observation-driven deterministic finite-state controller is a tuple
$
    \mathcal K=(S,s_0,\delta,\lambda),
$
where $S$ is finite, $s_0\in S$,
$\delta:S\times\Omega\to S$, and
$\lambda:S\to\mathsf{Act}$.  In state $s_t$, the controller emits
$a_t=\lambda(s_t)$.  If $a_t=\mathtt{halt}$, the computation
terminates.  If $a_t\in\mathcal C$, the environment executes the
corresponding tape command and returns the resulting observation
$\omega_t\in\Omega$.  If $a_t\in\Delta\cup\{\EOS\}$, the environment
appends this visible token to the output stream and returns the dummy
observation $\omega_t=\bot$.  In either nonhalting case, the controller
moves to
$
    s_{t+1}=\delta(s_t,\omega_t).
$
\end{definition}

\subsection{Selective affine SSM controllers}

\begin{definition}[Selective affine SSM controller]\label{def:sel-ssm}
A finite-precision selective affine SSM controller over \(\Omega\) is a tuple
$
    \mathcal M
    =
    \bigl(d,h^{(0)},\{A_\omega,b_\omega\}_{\omega\in\Omega},\rho\bigr),
$
where \(d\in\mathbb N\),
$
    h^{(0)}\in \mathbb F_p^d,
$
each
$
    A_\omega\in \mathbb F_p^{d\times d},
    b_\omega\in \mathbb F_p^d,
$
and
$
    \rho:\mathbb F_p^d\to \mathsf{Act}
$
is an action readout.
At time \(t\), the controller emits
$
    a_t=\rho(h^{(t)}).
$
If \(a_t=\mathtt{halt}\), the computation terminates. Otherwise the environment handles \(a_t\) according to the convention above: a tool command is executed by the tape oracle, while a visible output token is appended to the visible stream and produces the dummy observation \(\bot\). The returned observation \(\omega_t\in\Omega\) determines the hidden-state update
$
    h^{(t+1)}
    =
    A_{\omega_t}h^{(t)}
    +
    b_{\omega_t}.
$
This is the one-layer selective specialization of \((\mathrm{SSM})\), with \(B_{1,t}=0\), \(A_{1,t}=A_{\omega_t}\), and the output map given by the finite action readout \(\rho\).
\end{definition}

The controller is time-homogeneous in this specialization: the maps
$\{A_\omega,b_\omega\}_{\omega\in\Omega}$ and $\rho$ are fixed for
the entire computation.  They may depend on the current observation
and recurrent state, but not explicitly on the absolute time index.
The finite readout $\rho$ and the matrices are part of the fixed
description. The hidden dimension $d$ is the recurrent resource.

\subsection{Every finite controller is exactly realized by a selective SSM}

\begin{lemma}[Exact one-hot realization]
\label{lem:ssm-onehot}
Let
$
    \mathcal K=(S,s_0,\delta,\lambda)
$
be an observation-driven deterministic finite-state controller. Then there exists a finite-precision selective affine SSM controller that realizes exactly the same state evolution and emits exactly the same action at every time step.
One may choose hidden dimension
$
    d=|S|,
$
binary hidden states, transition matrices with entries in \(\{0,1\}\), and all biases equal to zero.
\end{lemma}

Here ``all biases equal to zero'' means that
\[
    b_\omega=0\in\mathbb F_p^d
    \qquad\text{for every }\omega\in\Omega.
\]
Thus, on each observation, the update is purely linear:
$h^{(t+1)}=A_{\omega_t}h^{(t)}$.

\begin{proof}
Enumerate the finite state set as
$
    S=\{s_1,\dots,s_m\},
    m=|S|.
$
Encode \(s_i\) by the \(i\)-th standard basis vector
$
    e_i\in \{0,1\}^m.
$
Let \(s_{i_0}=s_0\), and set
$
    h^{(0)}=e_{i_0}.
$

For each observation \(\omega\in\Omega\), define a matrix
$
    A_\omega\in\{0,1\}^{m\times m}
$
by specifying its columns:
\[
    A_\omega e_i=e_j
    \qquad
    \text{whenever }
    \delta(s_i,\omega)=s_j.
\]
Equivalently, the \(i\)-th column of \(A_\omega\) is the one-hot encoding of the state \(\delta(s_i,\omega)\). Set
\[
    b_\omega=0
    \qquad
    \text{for every }\omega\in\Omega.
\]

Define the SSM readout on reachable one-hot states by
$
    \rho(e_i)=\lambda(s_i).
$
On non-one-hot vectors, define \(\rho\) arbitrarily. Those vectors are never reached in the construction.

We prove by induction on \(t\) that, whenever the finite-state controller is in state \(s_i\) at time \(t\), the SSM hidden state is \(h^{(t)}=e_i\). The assertion holds at \(t=0\) by construction. Suppose it holds at time \(t\). Both systems emit the same action because
$
    \rho(h^{(t)})
    =
    \rho(e_i)
    =
    \lambda(s_i).
$
If this action is \(\mathtt{halt}\), both halt. Otherwise the environment returns the same observation \(\omega_t\) to both systems. If
$
    \delta(s_i,\omega_t)=s_j,
$
then the SSM update gives
$
    h^{(t+1)}
    =
    A_{\omega_t}h^{(t)}
    =
    A_{\omega_t}e_i
    =
    e_j,
$
which is exactly the one-hot encoding of the next finite state. Hence the induction closes.
Therefore the SSM emits exactly the same action sequence as \(\mathcal K\), and hence has the same visible output stream.
\end{proof}

\subsection{Turing completeness}

\begin{theorem}[Turing completeness of selective finite-precision SSM controllers]
\label{thm:ssm-turing-complete}
Let \(T\) be a deterministic single-tape Turing machine with finite state set \(Q\) and tape alphabet \(\Gamma\). Assume \(T\) is in output-normal form: whenever it halts on input \(x\), its tape contains
$
    \leftend\, T(x)\,\EOS
$
starting at the leftmost cell, and the head is positioned at the first symbol of \(T(x)\EOS\).
Then there exists a one-layer finite-precision selective affine SSM controller \(\mathcal M_T\) such that, with interactive access to the tape oracle, \(\mathcal M_T\) halts and emits
$
    T(x)\EOS
$
if and only if \(T\) halts on \(x\). If \(T\) does not halt on \(x\), then \(\mathcal M_T\) does not halt.
Moreover, the hidden dimension may be chosen to be
$
    d=O(|Q|\,|\Gamma|).
$
All transition matrices may be chosen with entries in \(\{0,1\}\), and all bias vectors may be chosen to be zero.
\end{theorem}

\begin{proof}
First construct an observation-driven finite-state controller
$
    \mathcal K_T=(S_T,s_0,\delta_T^{\mathcal K},\lambda_T)
$
that simulates \(T\) using the tape oracle.

The controller stores in finite control:
\begin{itemize}[leftmargin=2em]
\item the current simulated Turing-machine state \(q\in Q\);
\item a mode flag indicating whether it is about to read, write, move, or print;
\item when needed, one tape symbol \(\sigma\in\Gamma\) obtained from the most recent read.
\end{itemize}

Suppose the simulated machine is in a nonhalting state \(q\). The controller first emits the command \(\mathtt{read}\). If the returned observation is \(\sigma\in\Gamma\), it computes the Turing transition
\[
    \delta_T(q,\sigma)=(q',\sigma',d),
    \qquad
    d\in\{L,R,S\},
\]
in finite control. It then emits \(\mathtt{write}_{\sigma'}\). After receiving the dummy observation \(\bot\), it emits \(\mathtt{left}\), \(\mathtt{right}\), or \(\mathtt{stay}\), according as \(d=L\), \(d=R\), or \(d=S\). After the next dummy observation, it updates the stored simulated state to \(q'\) and repeats.

The standard simulation invariant is the following: after each completed simulated transition, the oracle tape configuration is exactly the tape configuration of \(T\), and the finite controller stores exactly the current finite control state of \(T\). The invariant holds at initialization by the definition of the tape oracle, and it is preserved because each simulated Turing step is implemented by exactly one read, one write, and one head-move command.

When the simulated machine reaches its halting state, the output-normal-form convention implies that the oracle head is positioned at the first symbol of \(T(x)\EOS\). The controller then enters print mode. In print mode it repeatedly emits \(\mathtt{read}\). If the returned symbol is \(\EOS\), it emits the visible token \(\EOS\) and then enters a state whose action is \(\mathtt{halt}\). If the returned symbol is \(\tau\in\Delta\), it emits the visible token \(\tau\), receives the dummy observation \(\bot\), emits \(\mathtt{right}\), receives \(\bot\), and repeats.

Thus \(\mathcal K_T\) emits precisely the visible stream \(T(x)\EOS\) when \(T\) halts, and it does not halt when \(T\) does not halt.
The number of states in \(\mathcal K_T\) is \(O(|Q|\,|\Gamma|)\): the controller needs only the simulated state, a constant-size mode flag, and at most one remembered tape symbol.

By Lemma~\ref{lem:ssm-onehot}, \(\mathcal K_T\) has an exact one-hot realization as a selective affine SSM controller of hidden dimension
$
    d=|S_T|=O(|Q|\,|\Gamma|).
$
The constructed SSM has binary transition matrices and zero biases. Since it emits exactly the same action sequence as \(\mathcal K_T\), it has exactly the same visible output stream and halting behavior. This proves the theorem.
\end{proof}

\begin{remark}[Dependence on selectivity]\label{rem:selectivity}
The theorem is a statement about selective finite-precision SSMs, where the transition matrix may depend on the current observation token. It should not be read as a statement about every narrower SSM architecture. In particular, a fixed linear time-invariant recurrence with one common transition matrix is not automatically able to implement arbitrary finite-state control.
\end{remark}

\section{Resource accounting: logarithmic and direct realization}\label{sec:resources}

The one-hot encoding of Theorem~\ref{thm:ssm-turing-complete} makes
the hidden dimension linear in $|Q|$. We now state two results,
proved in Appendix~\ref{app:log-simulation}, that answer distinct
resource questions under the conventions of
Section~\ref{sec:finite-state-controller}. Here $\mathcal T_B$ denotes
the languages over $\{0,1\}$ recognized by Turing machines with at
most $B$ internal states over the fixed tape alphabet
$\{0,1,\blank\}$, and $\mathcal S^{(1)}_b$ those recognized by a
one-layer selective affine controller with the tape oracle using at
most $b$ recurrent-state bits (Appendix~\ref{app:three-symbol}).

\begin{theorem}[Logarithmic tape-assisted simulation; restated from
Theorem~\ref{thm:log-simulation}]\label{thm:log-main}
There is an absolute constant $C_1$ such that
$\mathcal T_B\subseteq\mathcal S^{(1)}_{\lceil C_1\log_2(B+1)\rceil}$
for every $B\geq1$.
\end{theorem}

The construction hard-codes a self-delimiting description
$\langle M\rangle$ of length $O(B\log(B+1))$ of the simulated machine
in the fixed readout, copies it onto the tape using an $O(\log B)$-bit
counter, and hands control to a fixed universal tape controller built
from a small universal machine \citep{NearyWoods2009}. Only the
recurrent-state cost is logarithmic: the fixed description grows as
$O(B\log B)$, the occupied tape cells grow with the computation, and
no time bound is claimed.

\begin{theorem}[Direct realization; restated from
Theorem~\ref{thm:direct-realization}]\label{thm:direct-main}
Without external memory, let $D_L(B)$ be the worst-case minimum total
recurrent dimension needed to realize a map
$f:[B]\times\{0,1\}\to[B]$ exactly, one affine recurrent update per
input symbol, by an $L$-layer triangular selective-affine system.
Then $D_L(B)=B-1$ for every $B\geq2$ and every $L\geq1$. In
particular, $\Theta(B)$ recurrent bits are necessary at any fixed
precision, and exactly $B-1$ bits over $\mathbb F_2$.
\end{theorem}

\paragraph{Interpretation and alternative regimes.}
The two answers are compatible: the logarithmic bound uses external
persistent memory, many oracle interactions per simulated transition,
and a machine-specific readout, whereas the direct-realization bound
forbids all three. Their gap measures what external persistent
memory and machine-specific fixed parameters buy. If a tool
satisfying Definition~\ref{def:fs-oracle} has fixed finite capacity,
Proposition~\ref{prop:product} applies with an $|M|$-fold state
blow-up. For the local read/write/move tape of
Definition~\ref{def:tape}, usable capacity may grow with the
computation and Theorem~\ref{thm:tape-turing} applies. If parameter
or table size is charged, the one-hot construction remains an exact
expressivity result, but not a claim of description efficiency.

\section{Discussion and Conclusion}\label{sec:conclusion}

We established a separation between two regimes of tool access for
finite-precision recurrent sequence models. The first is a universal
internalization result: every finite-state bounded-interface tool can
be folded into the controller at a cost of at most
$\lceil\log_2|M|\rceil+O(1)$ recurrent bits, so the combined system
remains finite-state. The second is a sufficiency result: a single
unbounded, persistent, locally readable/writable tape enables
Turing-complete computation with $O(\log|Q|+\log|\Gamma|)$ bits of
recurrent memory, witnessed concretely by the exponential
$\mathrm{EQ}_n$ separation. The construction is realized exactly by a
one-layer selective affine SSM with binary one-hot dynamics, which
identifies \emph{selectivity} (observation-dependent transition
matrices) as an architectural ingredient that suffices for exact
finite control in the interactive setting. These results are not an
if-and-only-if taxonomy of infinite-state tools: counters, stacks,
write-only stores, and other restricted interfaces may have weaker,
intermediate, or interface-dependent power, and classifying them is
left to future work.

Section~\ref{sec:resources} and the supplementary material sharpen
the resource accounting: $O(\log B)$ recurrent bits suffice for
tape-assisted simulation of any $B$-state Turing machine when the
fixed description and occupied tape cells are reported separately,
while exactly $B-1$ recurrent dimensions are necessary and sufficient
for direct one-step realization without external memory. Both answers
are correct under their precisely specified conventions; the
tape-assisted result relocates persistent working memory to external
storage rather than eliminating it.

Our theorems concern exact expressivity under deterministic
finite-precision semantics. Learnability, stochastic or noisy tool
interaction, robustness, latency, and empirical performance of
contemporary agents are complementary questions that the present
results do not address. Further directions include learned and
stochastic controllers, transformers as controllers, multiple
concurrent tools, the formalization of practical tools (calculators,
retrieval, interpreters) as oracles with explicit states and
interfaces, and the interaction between tool access and
chain-of-thought budgets studied in \citet{merrill2024cot} and
\citet{zubic2026ssm}.


\bibliography{acml26}

\appendix
\newpage

\section{Realization with three-symbol model}\label{app:three-symbol}

Fix the input and tape alphabets
$
    \Sigma=\{0,1\}$, $\Gamma=\{0,1,\blank\}.
$
The tape is indexed by $\mathbb Z$.  On input
$x=x_0\cdots x_{n-1}$, cell $i$ initially contains $x_i$ for
$0\leq i<n$, every other cell contains $\blank$, and the head starts at
cell $0$.
The fixed tape oracle supports the commands
\[
    \readcmd,\quad
    \leftcmd,\quad
    \rightcmd,\quad
    \staycmd,\quad
    \writecmd{\gamma}\quad(\gamma\in\Gamma).
\]
A read command returns the currently scanned symbol in $\Gamma$.
Every write or move command returns the dummy observation $\ok$.
A deterministic Turing machine recognizes a language if it accepts
every member and never accepts a nonmember.  On a non-member, it may
reject or run forever.

\begin{definition}[State-bounded Turing-machine classes]
For $B\geq 1$, let
\[
    \mathcal T_B
    =
    \left\{
        K\subseteq\Sigma^*:
        K\text{ is recognized by a Turing machine with at most }
        B\text{ internal states}
    \right\}.
\]
Accept and reject are regarded as external halting outcomes and are
not counted as internal states.
\end{definition}

\begin{definition}[One-layer selective affine oracle controller]
A one-layer selective affine oracle controller consists of a recurrent
state
$
    h=(h_1,\ldots,h_d),
$
an action readout $\rho$, and, for every observation
$
    \omega\in\Omega:=\Gamma\cup\{\ok\},
$
an affine recurrent update:
$
    \Phi_\omega(h)=A_\omega h+b_\omega.
$
At each interaction step, the controller emits the action $\rho(h)$.
If this action is accepting or rejecting, the computation halts.
Otherwise, the tape oracle executes the action, returns an observation
$\omega$, and the state is updated to $\Phi_\omega(h)$.
Each recurrent coordinate is stored in a finite-precision register.
If coordinate $j$ uses $p_j$ bits, the recurrent-memory cost is:
\[
    \operatorname{mem}(\mathcal M)
    :=
    \sum_{j=1}^{d}p_j.
\]
The affine updates are required to be exact on every reachable state.
The fixed parameters and the finite readout table are outside the
recurrent-state metric and are reported separately as fixed
description size. Occupied tape cells are likewise reported
separately as external work memory.
Let $\mathcal S_b^{(1)}$ be the class of languages recognized by such
a controller, with access to the fixed tape oracle above, using at
most $b$ recurrent-state bits.
\end{definition}

The exclusion of parameter size from the recurrent-memory budget is
essential below.  In particular, the readout may contain a finite
lookup table depending on the simulated Turing machine.

\subsection{The linear one-hot simulation bound}

\begin{lemma}[Linear one-hot simulation bound]\label{lem:onehot-linear}
There is an absolute constant $C_0$ such that, for every $B\geq 1$,
\[
    \mathcal T_B
    \subseteq
    \mathcal S_{C_0B}^{(1)}.
\]
One may take, for example, $C_0=9$.
\end{lemma}

The $O(B)$ inclusion is also a direct consequence of the construction
in Theorem~\ref{thm:ssm-turing-complete}: once the tape alphabet is
fixed, its bound $O(|Q|\,|\Gamma|)$ becomes $O(B)$.  The proof below is
retained to make the controller-state bookkeeping and the explicit
constant $C_0=9$ transparent.

\begin{proof}
Let
$
    M=(Q,q_0,\delta)
$
be a Turing machine with $|Q|\leq B$ and tape alphabet $\Gamma$.
We first construct an observation-driven finite-state controller
$\mathcal K_M$ that simulates $M$ through the tape oracle.

For each $q\in Q$, introduce a state $R_q$ whose action is
$\readcmd$.  If the oracle returns $\gamma\in\Gamma$, the transition
table of $M$ determines either a halting outcome or a triple
\[
    \delta(q,\gamma)=(q',\tau,d),
    \qquad
    q'\in Q,\quad
    \tau\in\Gamma,\quad
    d\in\{L,R,S\}.
\]
In the latter case the controller passes through the following
constant-length protocol:
\[
    R_q
    \longrightarrow
    W_{q,\gamma}
    \longrightarrow
    D_{q,\gamma}
    \longrightarrow
    R_{q'}.
\]
The state $W_{q,\gamma}$ emits $\writecmd{\tau}$.
After receiving $\ok$, the state $D_{q,\gamma}$ emits
$\leftcmd$, $\rightcmd$, or $\staycmd$, according as
$d=L$, $R$, or $S$.  After the next observation $\ok$, the controller
enters $R_{q'}$.

If $\delta(q,\gamma)$ is accepting or rejecting, the controller enters
the corresponding halting state.  Thus the number of controller states
is at most:
\[
    |Q|+|\Gamma||Q|+|\Gamma||Q|+2
    \leq 7B+2
    \leq 9B.
\]

A direct induction on the number of simulated Turing-machine
transitions shows that, whenever the controller is in state $R_q$, the
oracle tape and head position coincide with those of $M$, and $q$ is
the current internal state of $M$.  Hence the controller recognizes
the same language as $M$.

It remains to realize this finite controller by a selective affine
SSM.  Enumerate its state set as:
\[
    S=\{s_1,\ldots,s_N\},
    \qquad
    N\leq 9B,
\]
and encode $s_i$ by the standard basis vector
$
    e_i\in\{0,1\}^{N}.
$
For each observation $\omega\in\Omega$, define a matrix
$A_\omega\in\{0,1\}^{N\times N}$ columnwise by
$
    A_\omega e_i=e_j
    \Longleftrightarrow
    \delta_{\mathcal K_M}(s_i,\omega)=s_j.
$
Set every bias to zero.  Define the action readout on the reachable
one-hot vectors by
$
    \rho(e_i)=\lambda_{\mathcal K_M}(s_i).
$
The set of one-hot vectors is invariant under every $A_\omega$, and
the resulting SSM exactly reproduces the controller's state and action
sequences.
Each recurrent coordinate is binary, so the recurrent-memory cost is
$N\leq 9B$ bits.  Therefore
$
    \mathcal T_B\subseteq\mathcal S_{9B}^{(1)}.
$
\end{proof}

\section{The logarithmic oracle-assisted simulation}\label{app:log-simulation}

The logarithmic statement becomes meaningful only after the oracle and
the resource accounting have been fixed.  Under the model above, it is
true.

\begin{lemma}[Fixed universal tape controller]
There exist:
\begin{enumerate}[label=(\roman*)]
\item a fixed observation-driven finite-state tape controller
      $\mathcal U$;
\item a self-delimiting binary encoding $M\mapsto\langle M\rangle$;
\end{enumerate}
such that, whenever the tape contains
$
    x\,\blank\,\langle M\rangle\,\blank
$
and the head is on the final blank, $\mathcal U$ accepts, rejects, or
runs forever exactly as $M$ does on $x$.
Moreover, if $M$ has at most $B$ internal states and tape alphabet
$\Gamma$, then
$
    |\langle M\rangle|
    \leq
    cB\log_2(B+1)
$
for an absolute constant $c$.
\end{lemma}

\begin{proof}
The existence of a fixed universal one-tape Turing machine with three
tape symbols follows, for example, by composing Theorems~2.1 and~3.1
of \citet{NearyWoods2009}.  Its three symbols may be
renamed as $0$, $1$, and $\blank$.

A fixed tape preprocessor can convert the layout
$
    x\,\blank\,\langle M\rangle\,\blank
$
into the particular initial encoding required by that universal
machine.  Composing the preprocessor with the universal machine and
then replacing every Turing transition by the fixed read--write--move
oracle protocol produces the required finite controller $\mathcal U$.

For the length estimate, label the states of $M$ by integers below
$B$.  Its transition table has $3B$ entries.  Each nonhalting entry
contains a next-state label of length
$\lceil\log_2 B\rceil$, a written symbol, and a head direction.
Halting entries require only constant additional information.  A
self-delimiting encoding of $B$, followed by the transition entries,
therefore has length
$
    O\!\left(B\log(B+1)\right).
$
\end{proof}

\begin{theorem}[Logarithmic oracle-assisted simulation]\label{thm:log-simulation}
There is an absolute constant $C_1$ such that, for every $B\geq 1$,
\[
    \mathcal T_B
    \subseteq
    \mathcal S_{
        \left\lceil C_1\log_2(B+1)\right\rceil
    }^{(1)}.
\]
\end{theorem}

\paragraph{Where the description of $M$ is stored.}
For each fixed simulated machine $M$, the finite word
$w=\langle M\rangle$ is hard-coded in the controller's fixed action
readout: when the counter has value $j$, the readout emits the
preassigned action $a_j$ that writes the next part of $w$.  Thus the
whole description is not initially stored in the recurrent state.
During the initialization phase, the controller copies $w$ onto the
oracle tape, which thereafter stores the description for the fixed
universal controller $\mathcal U$.  This is compatible with the
resource convention above because the readout table and the occupied
tape cells are outside the recurrent-state metric and are reported
separately (the readout has size $O(B\log B)$). Only the counter,
flag, and constant-size finite-control register are charged to
recurrent memory.

\begin{proof}
Fix a Turing machine $M$ with at most $B$ states.  Write
\[
    w=\langle M\rangle=w_1\cdots w_D,
    \qquad
    D\leq cB\log_2(B+1).
\]
We construct a one-layer selective affine controller that:

\begin{enumerate}[label=(\arabic*)]
\item scans to the first blank following the input;
\item writes $w$ immediately to the right of that blank;
\item transfers control to the fixed universal controller
      $\mathcal U$.
\end{enumerate}

Let
\[
    a_0,a_1,\ldots,a_{R-1},
    \qquad
    R=2D+1,
\]
be the fixed action script
\[
    \rightcmd,\,
    \writecmd{w_1},\rightcmd,\,
    \writecmd{w_2},\rightcmd,\,
    \ldots,\,
    \writecmd{w_D},\rightcmd.
\]
Every action in this script returns the dummy observation $\ok$.

Let $S_{\mathcal U}$ be the fixed finite state set of the universal
controller.  Add two scanner states:
\[
    s_{\mathrm{read}}
    \quad\text{and}\quad
    s_{\mathrm{move}}.
\]
Let
$
    u\in\{0,1\}^{k}
$
be a one-hot register for these states and the states of
$\mathcal U$, where $k$ is an absolute constant independent of $B$.
Let
$
    z\in\{0,1,\ldots,R\}
$
be an integer counter and let
$
    \eta\in\{0,1\}
$
be a program-writing flag.  The complete recurrent state is
$
    h=(z,\eta,u).
$

For every observation $\omega$, let $D_\omega$ be the zero--one
matrix implementing the appropriate one-hot transition of the fixed
scanner/universal controller.  In particular,
\[
    D_0e_{s_{\mathrm{read}}}
    =
    D_1e_{s_{\mathrm{read}}}
    =
    e_{s_{\mathrm{move}}},
\]
and
\[
    D_{\ok}e_{s_{\mathrm{move}}}
    =
    e_{s_{\mathrm{read}}}.
\]
Let $\ell(u)$ denote the coordinate of $u$ corresponding to
$s_{\mathrm{read}}$, and let $e_{\mathcal U,0}$ be the one-hot
encoding of the initial state of $\mathcal U$.

Define the recurrent updates on reachable states by
\[
    \Phi_{\ok}(z,\eta,u)
    =
    \bigl(z+\eta,\eta,D_{\ok}u\bigr),
    \tag{1}
\]
\[
    \Phi_{\blank}(z,\eta,u)
    =
    \bigl(
        0,\,
        \ell(u),\,
        D_{\blank}u+\eta e_{\mathcal U,0}
    \bigr),
    \tag{2}
\]
and, for $\sigma\in\{0,1\}$,
\[
    \Phi_\sigma(z,\eta,u)
    =
    \bigl(0,0,D_\sigma u\bigr).
    \tag{3}
\]
Each of these maps is affine.

The initial recurrent state is
$
    h^{(0)}=(0,0,e_{s_{\mathrm{read}}}).
$
The readout is defined as follows:
\begin{itemize}[leftmargin=2em]
\item in state $e_{s_{\mathrm{read}}}$ it emits $\readcmd$;
\item in state $e_{s_{\mathrm{move}}}$ it emits $\rightcmd$;
\item if $\eta=1$, $u=0$, and $z=j<R$, it emits $a_j$;
\item if $\eta=1$, $u=0$, and $z=R$, it emits $\readcmd$;
\item if $u$ encodes a state of $\mathcal U$, it emits the action
      prescribed by $\mathcal U$.
\end{itemize}
The readout may be defined arbitrarily on unreachable states.

The scanner alternates between reading the current tape cell and
moving right.  When it first reads $\blank$, equation~(2) sends
$
    (0,0,e_{s_{\mathrm{read}}})
    \longmapsto
    (0,1,0).
$
The head is then on the blank immediately following the input.

While $\eta=1$, every scripted action produces $\ok$, and equation~(1)
gives
\[
    (j,1,0)\longmapsto(j+1,1,0).
\]
Consequently, the controller emits the complete script and writes
$\langle M\rangle$ to the tape.  After the final move, the head is on
the blank immediately following the description and the recurrent
state is $(R,1,0)$.  Its next action is $\readcmd$.  The returned blank
and equation~(2) give
$
    (R,1,0)
    \longmapsto
    (0,0,e_{\mathcal U,0}).
$
At this point the tape has the form
$
    x\,\blank\,\langle M\rangle\,\blank,
$
the head is on the final blank, and the controller has entered the
initial state of $\mathcal U$.  It therefore accepts, rejects, or runs
forever exactly as $M$ does on $x$.

The counter $z$ requires
$
    \left\lceil\log_2(R+1)\right\rceil
    =
    O(\log(B+1))
$
bits.  The flag $\eta$ and the one-hot register $u$ require only a
constant number of additional bits.  Hence the total recurrent memory
is
$
    O(\log(B+1)).
$
All updates are time-homogeneous and affine. The machine-specific
description occurs only in the finite readout values assigned to the
counter states.
\end{proof}

\begin{remark}
The theorem relies on two conventions:

\begin{enumerate}[label=(\roman*)]
\item the tape oracle provides external persistent memory. Occupied
      tape cells are outside the recurrent-state metric and reported
      separately;
\item the finite model parameters and readout table are outside the
      recurrent-state metric and reported separately as fixed
      description size.
\end{enumerate}

If the simulated transition table had to be represented solely by the
recurrent affine state and applied in one recurrent step, the
logarithmic conclusion would be false.  This is precisely the issue
addressed next.
\end{remark}

\subsection{Direct realization of arbitrary transition maps}

The two questions referred to below are stated explicitly after the
model definition.  They are meaningful only after external memory is
excluded and ``implemented'' is required to mean one exact recurrent
update per input symbol.
We use a layered model that is at least as general as the usual
triangular selective-affine SSM recurrence.

\begin{definition}[Layered triangular selective-affine realization]
Fix a field $\mathbb K$ and $L\geq 1$.  An $L$-layer triangular
selective-affine recurrent system has states
\[
    h_\ell\in\mathbb K^{d_\ell},
    \qquad
    1\leq\ell\leq L.
\]
For each input symbol $a\in\{0,1\}$, the update has the form
\[
    h_1^+
    =
    A_{1,a}h_1+b_{1,a},
\]
and, for $2\leq\ell\leq L$,
\[
    h_\ell^+
    =
    A_{\ell,a}(h_1,\ldots,h_{\ell-1})h_\ell
    +
    b_{\ell,a}(h_1,\ldots,h_{\ell-1}).
    \tag{4}
\]
The coefficient functions in~(4) may depend arbitrarily on all lower
layers.  Dependence on updated lower-layer values is also covered,
since those values are themselves functions of the old lower-layer
prefix.
The system exactly realizes
\[
    f:[B]\times\{0,1\}\longrightarrow[B]
\]
if there is an injective encoding
\[
    e(q)
    =
    \bigl(e_1(q),\ldots,e_L(q)\bigr)
\]
such that
\[
    \Phi_a(e(q))=e(f(q,a))
\]
for every $q\in[B]$ and $a\in\{0,1\}$.

Define the worst-case required total recurrent dimension by:
\[
    D_L(B)
    :=
    \max_{f:[B]\times\{0,1\}\to[B]}
    \;
    \min
    \left\{
        \sum_{\ell=1}^{L}d_\ell:
        f\text{ has an exact }L\text{-layer realization}
    \right\}.
\]
\end{definition}

\paragraph{Question 1 (logarithmic direct realization).}
For a fixed number of layers $L$, is it true that
$
    D_L(B)=O(\log B)?
$
Equivalently, can every mapping
$f:[B]\times\{0,1\}\to[B]$ be implemented exactly, in one recurrent
update per input symbol and without external memory, using only
$O(\log B)$ recurrent-state coordinates (and hence $O(\log B)$ bits
at fixed precision)?

\paragraph{Question 2 (worst-case state requirement).}
If Question~1 has a negative answer, what is the asymptotically tight,
or ideally exact, growth of $D_L(B)$ as a function of $B$ and $L$?

\begin{definition}[Congruence of a deterministic transition system]
Let $Q$ be a finite state set with transition maps
$f_a:Q\to Q$ indexed by input symbols $a$.  A \emph{congruence} is an
equivalence relation $\sim$ on $Q$ that is preserved by every
transition:
\[
    q\sim q'
    \quad\Longrightarrow\quad
    f_a(q)\sim f_a(q')
    \qquad\text{for every }a.
\]
Equivalently, the transition maps descend to well-defined maps on the
quotient set $Q/{\sim}$.  The equality relation and the universal
relation, in which all states are equivalent, are always congruences.
\end{definition}

\begin{theorem}[Answer to Questions 1 and 2]\label{thm:direct-realization}
For every $B\geq 2$ and every $L\geq 1$,
\[
    D_L(B)=B-1.
\]
Consequently:

\begin{enumerate}[label=(\roman*)]
\item Question 1 has a negative answer, even if the number of layers
      is allowed to depend on $B$;
\item the asymptotically tight worst-case recurrent-state requirement
      is
      $
          \Theta(B)
      $
      bits under any fixed finite-precision convention;
\item over the Boolean field $\mathbb F_2$, where one scalar
      coordinate is one bit, the exact worst-case requirement is
      $B-1$ bits.
\end{enumerate}
\end{theorem}

\paragraph{Intuition for the lower bound.}
The hard transition system combines a cyclic permutation $c$ with a
``successor with an absorbing endpoint'' map $s$.  The cycle forces any
congruence classes to be invariant under cyclic translation, while the
map $s$ destroys every nontrivial translation-invariant partition.
Hence this automaton has only the universal congruence and equality.

In a triangular layered realization, equality of the first $\ell$
layer encodings defines a congruence, because the first $\ell$ updated
layers depend only on the first $\ell$ current layers.  Therefore, at
the first layer that distinguishes any states, all $B$ states must
become distinct at once.  The lower layers are constant there, so that
single layer must realize the entire length-$B$ successor chain by one
fixed affine map.  After translating the absorbing state to the origin,
this produces a nilpotent orbit of length $B-1$, which requires at least
$B-1$ linearly independent vectors and therefore dimension at least
$B-1$.

\begin{proof}
We first prove the upper bound
$
    D_L(B)\leq B-1.
$
It suffices to use one layer.

Index the state set by
$
    \{0,1,\ldots,B-1\}.
$
In $\mathbb K^{B-1}$, define
$
    e(B-1)=0
$
and
\[
    e(i)=v_i,
    \qquad
    0\leq i\leq B-2,
\]
where $v_0,\ldots,v_{B-2}$ are the standard basis vectors.

Fix an arbitrary map
\[
    f:\{0,\ldots,B-1\}\times\{0,1\}
      \longrightarrow
      \{0,\ldots,B-1\}.
\]
For each $a\in\{0,1\}$, set
$
    b_a=e(f(B-1,a))
$
and define the columns of $A_a$ by
\[
    A_av_i=e(f(i,a))-b_a,
    \qquad
    0\leq i\leq B-2.
\]
Then
$
    A_ae(i)+b_a=e(f(i,a))
$
for $0\leq i\leq B-2$, while
\[
    A_ae(B-1)+b_a=b_a=e(f(B-1,a)).
\]
Thus every $f$ has a one-layer affine realization in dimension
$B-1$.

We now prove the matching lower bound.  Identify the state set with
the cyclic group
\[
    \mathbb Z_B=\{0,1,\ldots,B-1\}.
\]
Consider the two transition maps
\[
    s(i)
    =
    \begin{cases}
        i+1,&0\leq i\leq B-2,\\
        B-1,&i=B-1,
    \end{cases}
    \tag{5}
\]
and
\[
    c(i)=i+1\pmod B.
    \tag{6}
\]
Define
\[
    f(i,0)=s(i),
    \qquad
    f(i,1)=c(i).
\]

We claim that the automaton $(\mathbb Z_B,s,c)$ has no congruences
other than equality and the universal relation.  Let $\sim$ be an
equivalence relation satisfying
\[
    i\sim j
    \quad\Longrightarrow\quad
    s(i)\sim s(j)
    \quad\text{and}\quad
    c(i)\sim c(j).
    \tag{7}
\]
Since $c^B$ is the identity, invariance under $c$ implies invariance
under every cyclic translation.  Hence the equivalence classes are
the cosets of the subgroup
$
    H:=\{h\in\mathbb Z_B:0\sim h\}.
$

Suppose that $H$ is nontrivial, and choose
$
    h\in H\setminus\{0\}.
$
Translation invariance gives
\[
    B-1\sim h-1\pmod B.
\]
Applying $s$ and using $h\neq0$, we obtain
\[
    B-1=s(B-1)\sim s(h-1)=h.
\]
Thus
\[
    h-(B-1)=h+1\pmod B
\]
belongs to $H$.  Since $h\in H$ and $H$ is a subgroup, it follows
that $1\in H$.  Therefore $H=\mathbb Z_B$.  Hence every congruence is
either equality or universal.

Now suppose that the pair $(s,c)$ has an $L$-layer realization
\[
    e(i)=\bigl(e_1(i),\ldots,e_L(i)\bigr).
\]
For each $\ell$, define an equivalence relation
\[
    i\equiv_\ell j
    \quad\Longleftrightarrow\quad
    \bigl(e_1(i),\ldots,e_\ell(i)\bigr)
    =
    \bigl(e_1(j),\ldots,e_\ell(j)\bigr).
\]
Because the first $\ell$ updated layers depend only on the first
$\ell$ current layers, $\equiv_\ell$ is a congruence of the automaton:
\[
    i\equiv_\ell j
    \quad\Longrightarrow\quad
    f(i,a)\equiv_\ell f(j,a)
\]
for $a\in\{0,1\}$.

The full encoding is injective, so $\equiv_L$ is equality.  Let $r$ be
the least layer for which $\equiv_r$ is equality.  Since there are no
nontrivial congruences,
$
    \equiv_{r-1}
$
is universal.  Consequently, the lower-layer prefix
$
    \bigl(e_1(i),\ldots,e_{r-1}(i)\bigr)
$
is independent of $i$, while $e_r$ itself is injective.

Under input $0$, the update of layer $r$, restricted to the encoded
states, is therefore one fixed affine map
$
    T(x)=Ax+b
$
satisfying
$
    T(e_r(i))=e_r(s(i)).
$

Let
$
    x_*:=e_r(B-1).
$
Because $s(B-1)=B-1$, we have
$
    T(x_*)=x_*.
$
Set
$
    y_i:=e_r(i)-x_*.
$
Then
\[
    Ay_i=y_{i+1}
    \qquad
    (0\leq i\leq B-2),
\]
where
$
    y_{B-1}=0.
$
In particular,
\[
    A^{B-1}y_0=0,
    \qquad
    A^{B-2}y_0=y_{B-2}\neq0.
\]

The vectors
$
    y_0,Ay_0,\ldots,A^{B-2}y_0
$
are linearly independent.  Indeed, suppose
$
    \sum_{j=0}^{B-2}\alpha_jA^jy_0=0
$
and let $j_0$ be the least index with $\alpha_{j_0}\neq0$.
Applying $A^{B-2-j_0}$ gives
$
    \alpha_{j_0}A^{B-2}y_0=0,
$
because all terms with index larger than $j_0$ acquire exponent at
least $B-1$ and vanish.  This contradicts
$A^{B-2}y_0\neq0$.
Thus layer $r$ has dimension at least $B-1$, and hence
$
    \sum_{\ell=1}^{L}d_\ell\geq B-1.
$
Combining this with the upper bound proves
$
    D_L(B)=B-1.
$

At fixed scalar precision, dimension and recurrent bit complexity
differ by only a constant factor, giving the tight bound
$\Theta(B)$.  Over $\mathbb F_2$, the simplex construction uses
exactly $B-1$ binary coordinates, so the exact bit requirement is
$B-1$.
\end{proof}

\begin{remark}[Why the two results are compatible]
The logarithmic oracle-assisted theorem and the linear direct
realization lower bound answer different questions.

The oracle-assisted simulation is allowed to:
\begin{itemize}[leftmargin=2em]
\item use the unbounded tape as external persistent memory;
\item write the simulated machine's transition table to the tape;
\item take arbitrarily many oracle interactions per simulated
      transition;
\item encode the machine-specific program in the fixed readout
      parameters.
\end{itemize}

By contrast, the direct realization problem requires the mapping
$
    (q,a)\longmapsto f(q,a)
$
to be performed in one recurrent update, without external memory.
For this one-step affine problem, $\Theta(B)$ recurrent state is both
necessary and sufficient.
\end{remark}

\end{document}